\documentclass[11pt,reqno]{amsart}

\usepackage{amsmath, amsfonts, amssymb,  mathrsfs,  array, stmaryrd,  indentfirst, amsthm,  hyperref, comment}
\usepackage{graphicx, enumitem, tabularx, color}
\usepackage{cases}
\usepackage{apacite}
\usepackage[margin=1.0in]{geometry}
\usepackage{setspace}
\numberwithin{equation}{section}
\theoremstyle{plain}
\newtheorem{lemma}{Lemma}[section]
\newtheorem{thm}{Theorem}[section]
\newtheorem{defn}{Definition}[section]
\newtheorem{prop}{Proposition}[section]
\newtheorem{cor}{Corollary}[section]
\newtheorem{assumption}{Assumption}[section]
\newtheorem{remark}{Remark}[section]
\newtheorem{Example}{Example}[section]

\newcommand{\E}{\mathbb{E}}

\newcommand{\PP}{{\mathbb{P}}}

\newcommand{\eps}{{\varepsilon}}

\makeatletter
\def\@setcopyright{}
\def\serieslogo@{}
\makeatother

\begin{document}

\title[]{Equilibrium concepts for time-inconsistent stopping problems in continuous time}
\thanks{E. Bayraktar is supported in part by the National Science Foundation under grant DMS-1613170 and by the Susan M. Smith Professorship.}

\author{Erhan Bayraktar}
\address[Erhan Bayraktar]{Department of Mathematics, University of Michigan.}
\email{erhan@umich.edu}
\author{Jingjie Zhang }
\address[Jingjie Zhang]{Department of Mathematics, University of Michigan.}
\email{jingjiez@umich.edu}
\author{Zhou Zhou}
\address[Zhou Zhou]{School of Mathematics and Statistics, University of Sydney.}
\email{zhou.zhou@sydney.edu.au}

\begin{abstract}
A \emph{new} notion of equilibrium, which we call \emph{strong equilibrium}, is introduced for time-inconsistent stopping problems in continuous time. Compared to the existing notions introduced in \cite{Decreasing} and \cite{Christensen1}, which in this paper are called \emph{mild equilibrium} and \emph{weak equilibrium} respectively, a strong equilibrium captures the idea of subgame perfect Nash equilibrium more accurately. When the state process is a continuous-time Markov chain and the discount function is log sub-additive, we show that an optimal mild equilibrium is always a strong equilibrium. Moreover, we provide a new iteration method that can directly construct an optimal mild equilibrium and thus also prove its existence.
\end{abstract}


\maketitle

\tableofcontents

\section{Introduction}
On a filtered probability space $(\Omega,\mathcal{F},(\mathcal{F}_t)_{t\in[0,\infty)},\mathbb{P})$ consider an optimal stopping problem in continuous time
\begin{equation}\label{P1}
    \sup_{\tau \in \mathcal{T}} \E_x[\delta(\tau) X_{\tau}],
\end{equation}
where $X= (X_t)_{t\in [0, \infty)}$ is a time-homogeneous Markov process taking values in some space $\mathbb{X}\subset \mathbb{R}$, $\mathcal{T}$ is a set of stopping times, $\delta$ is a discount function, and $\E_x$ is the expectation given $X_0=x$. 
It is well known that when $\delta$ is not exponential, the problem \eqref{P1} may be time-inconsistent. That is,  the optimal stopping strategy obtained today may not be optimal in the eyes of future selves.  { There are mainly three ways to approach such time inconsistency: pre-committed strategy, naive strategy and consistent planning strategy, dating back to   \cite{Strotz}. Our paper focuses on consistent planning, which is formulated as a subgame perfect Nash equilibrium: once an equilibrium strategy is enforced over the planning horizon, the current self has no incentive to deviate from it, given all future selves will follow the equilibrium strategy. For discussions on different approaches, see \cite{Pollak}, \cite{Ekeland}, \cite{B2014}, \cite{B2017}, \cite{Christensen1}, \cite{ProbabilityDistortion}, \cite{Decreasing} and references therein. }

There are two general notions of equilibrium stopping strategies in continuous time in the literature. The first notion is proposed in \cite{Decreasing} and further studied in \cite{ProbabilityDistortion,HZOptimalContinous,2019arXiv190601232H}, which we will call mild equilibrium in this paper.  Following \cite[Definition 3.3]{Decreasing} and \cite[Definition 2.2]{HZOptimalContinous}, we have the following definition of mild equilibrium.

\begin{defn}\label{d1}
A measurable set $S \subset \mathbb{X}$ is said to be a mild equilibrium, if
\begin{numcases}{}
\label{e1} x \le \E_x[\delta(\tau_S)X_{\tau_S}], \quad \forall\,x\notin S,\\
\label{ee2} x\geq \E_x[\delta(\tau_S^+)X_{\tau_S}], \quad \forall\,x\in S,
\end{numcases}
where
\begin{equation}\label{e4}
\tau_S:= \inf\{t \ge 0: X_t \in S\},\quad\text{and}\quad \tau_S^+:= \inf\{t > 0: X_t \in S\}.
\end{equation}
\end{defn}
In the above $S$ is the stopping region, and the economic interpretation for Definition \ref{d1} is clear, there is no incentive to deviate. That is, in \eqref{e1} when $x\notin S$, it is better to continue and get $\E_x[\delta(\tau_S)X_{\tau_S}]$, rather than to stop and get $x$; \textit{on the surface} a similar statement applies to \eqref{ee2}. However, when the time of return for $X$ is $0$ (i.e., { $\mathbb{P}(\tau^+_{\{x\}}=0\,|\,X_0=x)=1$ }), which is satisfied for continuous-time Markov chain and many one-dimension diffusion processes, $\tau_S=\tau_S^+$ and thus \eqref{ee2} trivially holds. In other words, when the time of return is 0, there is no actual deviation captured by \eqref{ee2} from stopping to continuing, and Definition \ref{d1} is equivalent to the following.

\begin{defn}\label{Def_mild}
A measurable set $S \subset \mathbb{X}$ is said to be a mild equilibrium, if
\begin{equation}\label{e5}
x \le \E_x[\delta(\tau_S)X_{\tau_S}]=:J(x,S), \quad \forall\,x\notin S.
\end{equation}
\end{defn}

Consequently, with the time of return being $0$ the notion of mild equilibrium cannot fully capture the economic meaning of equilibrium. It is easy to see that $S=\mathbb{X}$ is always a mild equilibrium, and it is not clear why always stopping immediately is a reasonable strategy. { Notice that in discrete time there is no such degeneration issue for equilibrium since $\tau_S^+= \inf\{t \ge 1: X_t \in S\}$ in discrete time setting. See \cite[Remark 2.3]{HZOptimalDiscrete} and \cite[Definition 2.2]{doi:10.1137/18M1216432}.}

As can be seen from \cite{Decreasing,ProbabilityDistortion,HZOptimalContinous,2019arXiv190601232H}, there is often a continuum of mild equilibria in many natural models, which naturally leads to the question of equilibrium selection. In \cite{HZOptimalContinous}, optimal mild equilibrium in the sense of point-wise dominance is considered. In particular, from \cite[Definition 2.3]{HZOptimalContinous} we have the following definition.

\begin{defn}A mild equilibrium $S^*$ is said to be optimal, if for any other mild equilibrium $S$,
$$x\vee J(x, S^*) \ge x\vee J(x, S) (\Longleftrightarrow J(x, S^*)\geq J(x,S)),\quad\forall\,x\in\mathbb{X}.$$
\end{defn}

Note that $x\vee J(x, S)$ represents the value associated with the stopping region/strategy $S$. In \cite{HZOptimalContinous} the existence of optimal equilibrium is established. A discrete-time version is in \cite{HZOptimalDiscrete}. 

The second notion of equilibrium, which we call weak equilibrium in this paper, is proposed in \cite{Christensen1} and further investigated in \cite{Christensen2}. Following \cite{Christensen1}, we have the definition of weak equilibrium (we adapt the definition slightly for our setting).
\begin{defn} \label{Def_weak} A measurable set $S \subset \mathbb{X}$ is said to be a weak equilibrium, if
\begin{numcases}{}
x \le \E_x[\delta(\tau_S)X_{\tau_S}], \quad \forall  x\notin S,\notag\\
\liminf_{\varepsilon \searrow 0}\frac{x - \E_x[\delta(\tau_S^{\varepsilon}) X_{\tau_S^{\varepsilon}}]}{\varepsilon} \ge 0, \quad \forall x \in S,\label{ii}
\end{numcases}
where
\begin{equation}\label{e6}
\tau_S^{\varepsilon} = \inf\{t\ge \varepsilon: X_t \in S\}.
\end{equation}
\end{defn}

Compared to \eqref{ee2}, the first-order condition \eqref{ii} does capture the deviation from stopping to continuing. However, similar to that for time-inconsistent control (see e.g., \cite[Remark 3.5]{B2017} and \cite{StrongWeak}), the first-order criterion does not correspond to the equilibrium concept perfectly: when the limit in \eqref{ii} equals zero, it is possible that for all $\varepsilon> 0,$ $x < \E_x[\delta(\tau_S^{\varepsilon}) X_{\tau_S^{\varepsilon}}]$, in which case there is an incentive to deviate.

To sum up, the economic interpretation of being ``equilibrium'' for mild and weak ones is inadequate. { There are similar issues in continuous-time time-inconsistent stochastic control problems as mentioned in \cite[Remark 3.5]{B2017}. In response to \cite[Remark 3.5]{B2017}, a new definition of continuous-time equilibrium control is introduced in \cite{StrongWeak}. 
In time-inconsistent optimal stopping problems,  we introduce the following concept of strong equilibrium, which is inspired by \cite{StrongWeak}. }
\begin{defn} \label{Def_Strong}A measurable set $S \subset \mathbb{X}$ is said to be a strong equilibrium, if
\begin{numcases}{}
x \le \E_x[\delta(\tau_S)X_{\tau_S}], \quad \forall  x\notin S,\notag\\
{
\exists\,\, \varepsilon=\varepsilon(x)>0, \text{s.t. } \forall\,\varepsilon'\in(0,\varepsilon),\  x \ge \E_x[\delta(\tau_S^{\varepsilon'}) X_{\tau_S^{\varepsilon'}}], \quad \forall x \in S.
}\label{ii2}
\end{numcases}
\end{defn}

Compared to \eqref{ee2} and \eqref{ii}, condition \eqref{ii2} not only captures the deviation from stopping to continuing, but also more precisely indicates the disincentive of such deviation. Consequently, a strong equilibrium delivers better economic meaning as being an ``equilibrium''.

In this paper, when $X$ is a Markov chain we show that an optimal mild equilibrium is a strong equilibrium (see Theorem \ref{Connection}). (Obviously, a strong equilibrium is also weak, and a weak equilibrium is also mild.)  We also provide examples showing that a strong equilibrium may not be an optimal mild equilibrium, and a weak equilibrium may not be strong. Therefore, we thoroughly obtain the relation between mild, weak, strong, and optimal mild (and thus optimal weak, optimal strong) equilibria. Moreover, we provide a new iteration method which directly constructs an optimal mild equilibrium and thus also establish its existence (see Thoerem \ref{Existence}). In \cite{HZOptimalDiscrete,HZOptimalContinous}, an optimal equilibrium is constructed by the intersection of all (mild) equilibriums. In principle, this requires us to first find all (mild) equilibria in order to get the optimal one, which may not be implementable in many cases. The new iteration method proposed in this paper is much easier and more efficient to implement. { Examples are provided to demonstrate the application of the new iteration method (see Example 3.1 and Example 3.2).} It would be interesting to see whether such results can be extended to diffusion models, which we will leave for future research.

As in reality people often discount non-exponentially, our results can be applied to stopping problems in finance ad economics. { Generally we can use $X = f(Y)$ for some nonnegative payoff function $f$ and some price process of underlying asset $Y$. Our results still hold and the proofs still work when replacing $X$ with $f(Y)$. For instance, in Example 3.2, $Y$ is a stock price process and  $X = f(Y)$ is the payoff of an American put option. This can be viewed as an example of exercising an American option when the investor tries to maximize the expected payoff yet subject to hyperbolic discounting. We refer to \cite[Section 5]{HZOptimalDiscrete} and \cite[Section 6.3]{HZOptimalContinous} for more such examples. }The two-state example provided in Section 4 of this paper can also be thought of as an application of stopping (e.g., selling a house) when the economy (e.g., property market) is good/bad. Our paper is inline with the work \cite{GW07}, where equilibrium stopping strategies are considered in an entrepreneur's investment-timing problem under time-inconsistent preferences due to quasi-hyperbolic discounting. People have considered to incorporate non-exponential discounting into decision making including optimal stopping. However, this leads to time inconsistency as argued by Grenafier and Wang in \cite{GW07}. They proposed the time-consistent modeling framework and our result can be seen as making advances on proposing better equilibrium concepts in this line of work.
Let us also mention that the notion of strong equilibria also applies to other types of time-inconsistent stopping, such as mean-variance stopping problems and stopping under probability distortion.\footnote{For instance, consider $G(x,\tau)$, where $x$ is the initial position for the underlying process $X$, and $G$ is payoff utility. For example, $G(x,\tau)=\E_x[\delta(\tau)f(X_\tau)]$ for stopping with non-exponential discounting, $G(x,\tau)=\E_x[f(X_\tau)]-c\text{Var}_x[f(X_\tau)]$ for mean-variance stopping, $G(x,\tau)=\int_0^\infty w(\mathbb{P}[f(X_\tau^x)>y])dy$ for stopping under probability distortion.
Note that $G(x,0)=f(x)$. Then in general strong equilibria can be formulated accordingly as: $S \subset \mathbb{X}$ is said to be a strong equilibrium, if
\begin{numcases}{}
\notag f(x) \le G(x,\tau_S), \quad \forall  x\notin S,\notag\\
\notag  {  \exists\,\,\varepsilon=\varepsilon(x) >0, \text{s.t. } \forall\,\varepsilon'\in(0,\varepsilon),\  f(x) \ge G(x,\tau_S^{\varepsilon'})}, \quad \forall x \in S.
\end{numcases}
The mild and weak equilibria can also be defined accordingly, and they still suffer from being short of economic meaning.
}

{ This paper provides very novel and conceptual contributions in the topic of time-inconsistent stopping. First, we analyze existing notions of equilibrium and their inadequacy in continuous time. A new notion of equilibrium, strong equilibrium is introduced. It captures the economic meaning of being ``equilibrium" more accurately. Second, we show that an optimal mild equilibrium is also a strong equilibrium, which is far from obvious. This result together with the examples in this paper completely shows the relations between mild, weak, strong, optimal mild/weak/strong equilibria. No such result has been obtained before. Moreover, we completely obtain the existence and (non)uniqueness results of these equilibria.  Third, we provide an iteration method, which directly constructs an optimal equilibrium and is much more implementable than the existing method in \cite{HZOptimalDiscrete,HZOptimalContinous}. Moreover, although the proofs are relatively short, by no means they are trivial, routine or easy to come up with. Those key ideas provide some novel proof approaches in the literature of time-inconsistent control/stopping. Let us mention the recent work of \cite{https://dx.doi.org/10.2139/ssrn.3308274}, where the authors also discuss notions of equilibrium control based on condition (1.3) in \cite{StrongWeak}. The focus of their paper is to distinguish between weak and strong (and regular) equilibrium controls. It is intuitively like distinguishing between local maxima and critical point.  In our paper we not only distinguish between strong, weak and mild equilibrium stopping times, but also obtain that an optimal mild equilibrium is a strong equilibrium. The notions of mild equilibrium and optimal mild equilibrium only make sense for stopping problems not control problems. Thus the focus of our paper is different from  \cite{https://dx.doi.org/10.2139/ssrn.3308274} and our intuitively unexpected result makes a novel contribution to the literature.
}

The rest of the paper is organized as follows. Section 2 collects the main results of the paper. An optimal mild equilibrium is proved to be a strong equilibrium, and can be directly constructed via a new iteration method.  {Section 3 provides examples to illustrate the iteration method in Theorem \ref{Existence}.} Section 4 focuses on a concrete two-state model, which demonstrates the differences between these equilibria.
 
\section{ The Main Results}
In this section, we apply the concepts in Section 1 to a continuous-time Markov chain and  present our main results under this setting. Let $X = (X_t)_{t\ge 0}$ be a time-homogeneous continuous-time Markov chain. It has a finite or countably infinite state space $\mathbb{X}\subset[0,\infty)$. Let $\lambda_x$ be the transition rate out of the state $x \in \mathbb{X}$, and $q_{xy}$ be the transition rate from state $x$ to $y$ for $y \ne x$. Then we have that $\lambda_x = \sum_{y \ne x} q_{xy}$. The discount function $t\mapsto\delta(t)$ is assumed to be non-exponential and decreasing, with $\delta(0)=1$ and $\lim_{t\to\infty}\delta(t)=0$. Let the filtration $(\mathcal{F}_t)_{t\in[0,\infty)}$ be generated by $X$. Furthermore, we make the following assumptions on $X$ and $\delta(\cdot)$.

\begin{assumption}\label{A1}
(i) $C:=\sup{\mathbb{X}}<\infty$ and $\lambda := \sup_{x\in \mathbb{X}}\lambda_x < \infty$.
 
(ii) $X$ is irreducible, i.e., for any $x ,y \in \mathbb{X}, \,\, \inf\{t\ge 0: X_t = y\,|\,X_0 =x\} < \infty$, a.s.. 
\end{assumption}

\begin{assumption} \label{A2}
(i) $\delta$ is log-subadditive, i.e.,
\begin{equation}\label{ee3}
\delta(s)\delta(t) \le \delta(s+t), \quad \forall\,s, t > 0.
\end{equation}

(ii) $t\mapsto\delta(t)$ is differentiable at $t=0$, and $\delta'(0)<0$.
\end{assumption}
 
 \begin{remark}
Assumption \ref{A2} (i) is closely related to {\it decreasing impatience} ($DI$) in Behavioral Finance and Economics. \footnote{\color{black}As mentioned in \cite{Decreasing}: ``It is well-documented in empirical studies, e.g. \cite{LP92,LT89,Thaler81}, that people admits $DI$: when choosing between two rewards, people are more willing to wait for the larger reward (more patient) when these two rewards are further away in time. For instance, in the two scenarios (i) getting \$100 today or \$110 tomorrow, and (ii) getting \$100 in 100 days or \$110 in 101 days, people tend to choose \$100 in (i), but \$110 in (ii)."}
Following \cite[Definition 1]{Prelec04} and \cite{Noor2009a}, the discount function $\delta$ induces $DI$ if
\begin{equation}\label{DI def}
s\mapsto \frac{\delta(s+t)}{\delta(s)}\ \hbox{is strictly increasing,}\quad\forall\,t>0.
\end{equation}
Observe that \eqref{DI def} implies \eqref{ee3}, since $\delta(s + t)/\delta(s)\ge \delta(t)/\delta(0) = \delta(t)$ for all $s,t\ge 0$.

Note that hyperbolic, generalized hyperbolic, quasi-hyperboic, pseudo-exponential discount functions all induce DI, and thus satisfy Assumption \ref{A2} (i). Consequently, \eqref{ee3} is often used when studying problems involving non-exponential discounting; see e.g., \cite{Decreasing,HZOptimalDiscrete,HZOptimalContinous}. 
 \end{remark}
 
 The following is the first main result of this paper, which shows that an optimal mild equilibrium is a strong equilibrium. The proof is provided in Section 2.1.
  \begin{thm}\label{Connection}
 Let Assumptions \ref{A1} and \ref{A2} hold. If $S$ is an optimal mild equilibrium, then it is a strong equilibrium.
 \end{thm}
Since all mild equilibria are strong equilibria, an optimal mild equilibrium will generate larger values than any strong equilibrium as well. With Theorem \ref{Connection}, we can conclude that any optimal mild equilibrium is a strong equilibrium and in fact is an optimal strong equilibrium.

The following is the second main result of this paper. It provides an iteration method which directly constructs an optimal mild equilibrium, and thus also establishes the existence of weak, strong, and optimal mild equilibria. The proof of this result is presented in Section 2.2.

\begin{thm}\label{Existence}
Let $S_0 := \emptyset$, and
\begin{equation}\label{e8}
S_{n+1} := S_n \cup \left\{x\in\mathbb{X}\setminus S_n: x > \sup_{S: S_n \subset S \subset \mathbb{X}\backslash \{x\}} J(x, S)\right\}.
\end{equation}
Let
\begin{equation}\label{e7}
S_{\infty}:=\cup_{n=0}^\infty S_n.
\end{equation}
If Assumptions \ref{A1} (i) and \ref{A2} (i) hold, then $S_\infty$ is an optimal mild equilibrium. If in addition Assumption \ref{A2} (ii) holds, then $S_\infty$ is a strong equilibrium.
\end{thm}

 \subsection{Proof of Theorem \ref{Connection}}
 
Recall $\tau_S,\tau_S^\eps,J(\cdot,\cdot)$ defined in \eqref{e4},\eqref{e6},\eqref{e5} respectively. We have the following  characterization of \eqref{ii} in Definition \ref{Def_mild}.
\begin{prop}\label{p1}
Let Assumptions \ref{A1} and \ref{A2} (ii) hold. Then $S \subset \mathbb{X}$ is a weak equilibrium if and only if $S$ is a mild equilibrium and for all $x\in S$,\[
x(\lambda_x - \delta'(0)) \ge \sum_{y\in S\backslash \{x\}}y q_{xy} + \sum_{y \in S^c} J(y,S)q_{xy}.
\]
\end{prop}
 
 \begin{proof} By definition, we only need to check condition (\ref{ii}) in Definition \ref{Def_mild} is equivalent to the above inequality.
 
Denote $T_x := \inf\{t\ge 0: X_t \ne x, X_0 = x\}$ as the holding time at state $x$, which has exponential distribution with parameter $\lambda_x$.
Then \begin{align*}
\E_x[\delta(\tau_S^{\varepsilon}) X_{\tau_S^{\varepsilon}}]=\,\, &\E_x[\delta(\tau_S^{\varepsilon})X_{\tau_S^{\varepsilon}}\textbf{1}_{\{T_x > \varepsilon\}}] + \sum_{y \in  \mathbb{X}\backslash\{x\}}\E_x[\delta(\tau_S^{\varepsilon}) X_{\tau_S^{\varepsilon}}\textbf{1}_{\{T_x \le \varepsilon,X_{T_x} = y, T_y+T_x > \varepsilon\}}]+ O(\varepsilon^2)\\
 =\,\, &\delta(\varepsilon)x e^{-\lambda_x \varepsilon} + \left\{\sum_{y \in  S\backslash\{x\}}\delta(\varepsilon)y\frac{q_{xy}}{\lambda_x} + \sum_{y \in  S^c}\E_y[\delta(\varepsilon+\tau_S)X_{\tau_S}]\frac{q_{xy}}{\lambda_x}\right\}(\lambda_x \varepsilon +O(\varepsilon^2)) + O(\varepsilon^2). 
 \end{align*}
 
 Notice that { $\delta(\varepsilon) = 1+  \delta'(0)\varepsilon +o(\varepsilon)$}. Therefore we have \[
 \E_x[\delta(\tau_S^{\varepsilon}) X_{\tau_S^{\varepsilon}}] = x+ \left\{-x(\lambda_x -\delta'(0)) +\sum_{y\in S\backslash\{x\}}y q_{xy}+\sum_{y \in S^c}q_{xy}\E_y[\delta(\varepsilon+\tau_S)X_{\tau_S}]\right\} \varepsilon  +o(\varepsilon).
 \]
Therefore, \eqref{ii} is equivalent to
\[
 x(\lambda_x - \delta'(0)) \ge \sum_{y\in S\backslash \{x\}}y q_{xy} + \sum_{y \in S^c} \E_y[\delta(\tau_S)X_{\tau_S}]q_{xy}.
 \]
 \end{proof}

\begin{cor} \label{rmk2}
Let Assumptions \ref{A1} and \ref{A2} (ii) hold. If $S$ is a mild equilibrium and satisfies
\[
 x(\lambda_x - \delta'(0)) > \sum_{y\in S\backslash \{x\}}y q_{xy} + \sum_{y \in S^c} \E_y[\delta(\tau_S)X_{\tau_S}]q_{xy},
 \]
 then it is a strong equilibrium.
\end{cor}

For the rest of the paper, we will sometimes use the notation
$$\rho(x,S):=\inf\{t\geq 0:\ X^x_t\in S\}$$
in the place of $\tau_S$ to emphasize the initial state $X_0=x$ ($X^x$ here is the Markov chain starting at $x$).

\begin{lemma} \label{L11}
Let Assumption \ref{A2} (i) hold. For $x\in S$, denote $\hat S = S\backslash \{x\}$. If $S$ is an optimal mild equilibrium, then for any $y \notin S$, \[
 J(y, \hat S) - J(y, S)\ge \E_y[\delta(\tau_{S})\textbf{1}_{\{X_{\tau_S} = x\}}](J(x,\hat S) -x).
 \]
 \end{lemma}
 
 \begin{proof} Since $\hat S \subset S$, we have $\rho(y, S) \le \rho(y, \hat S)$. Then
 \begin{align*}
 &J(y, \hat S) - J(y, S)\\
 &=\E_y[\delta(\rho(y,\hat S))X_{\rho(y,\hat S)}\textbf{1}_{\{X_{\rho(y,S)} = x\}}] + \E_y[\delta(\rho(y,\hat S))X_{\rho(y,\hat S)}\textbf{1}_{\{X_{\rho(y,S)} \in \hat S\}}] -  \E_y[\delta(\rho(y,S))X_{\rho(y, S)}] \\
&=\E_y[\delta(\rho(y,\hat S))X_{\rho(y,\hat S)}\textbf{1}_{\{X_{\rho(y,S)} = x\}}] + \E_y[\delta(\rho(y,S))X_{\rho(y,S)}\textbf{1}_{\{X_{\rho(y,S)} \in \hat S\}}] -  \E_y[\delta(\rho(y,S))X_{\rho(y, S)}] \\
&= \E_y[\delta(\rho(y,\hat S))X_{\rho(y,\hat S)}\textbf{1}_{\{X_{\rho(y,S)} = x\}}]  -  x\E_y[\delta(\rho(y,S))\textbf{1}_{\{X_{\rho(y,S)} = x\}}]  \\
&\geq \E_y[\delta(\rho(y,S))\textbf{1}_{\{X_{\rho(y,S)} = x\}}\E[\delta(\rho(y,\hat S)-\rho(y,S))X_{\rho(x,\hat S)}|\mathcal{F}_{\rho(y,S)}]] -  x\E_y[\delta(\rho(y,S))\textbf{1}_{\{X_{\rho(y,S)} = x\}}]  \\
&= \E_y[\delta(\tau_{S})\textbf{1}_{\{X_{\tau_S} = x\}}](\E_x[\delta(\tau_{\hat S})X_{\tau_{\hat S}}] -x),
  \end{align*}
 where we use \eqref{ee3} for the inequality above. 
 \end{proof}

\begin{lemma}\label{L12}
Let Assumption \ref{A2} (i) hold. If $S$ is an optimal mild equilibrium, then for any $x \in S$ we have that
\[
 x\ge J(x, \hat S), \quad \text{where} \quad \hat S = S\backslash\{x\}.
 \]
 As a result, $0 \notin S$ and $J(y, S)> 0$ for all $y \in \mathbb{X}$.
 \end{lemma}

 \begin{proof}
 
 If $\hat S$ is also a mild equilibrium, then
 $$x \leq J(x, \hat S)\leq J(x,S)=x,$$
 and thus $x=J(x,\hat S)$.
 
 If $\hat S$ is not a mild equilibrium, then there exists $y \notin \hat S$ such that $J(y, \hat S) < y \le J(y, S)$. By Lemma \ref{L11}, \[
 0 > J(y, \hat S) - J(y, S)\ge  \E_y[\delta(\tau_{S})\mathbb{I}_{\{X_{\tau_S} = x\}}](J(x,\hat S) -x),
 \]
which implies that
\begin{equation}\label{e11}
x>J(x,\hat S).
\end{equation}
 
Now suppose $0 \in S$. By the above result, we have $0 \ge J(0, S\backslash\{0\})$. Since $X_{\tau_{S\backslash\{0\}}} >0$, $J(0, S\backslash\{0\})> 0$, which is a contraction. As a result, $0 \notin S$ and $J(y, S)>0$ for all $y \in \mathbb{X}$.
 \end{proof}

 \begin{proof}[\textbf{Proof of Theorem \ref{Connection}}]
  By Assumption \ref{A2}, $\delta(t) \ge e^{\delta'(0)t}$ for all $t\ge 0$. Moreover, there exist $t_0>0$ such that for $t> t_0, \,\, \delta(t) > e^{\delta'(0)t}$ since $\delta$ is non-exponential. As a result, for any $x\in \mathbb{X}$,
\[
 \E_x[\delta(T_x)] = \int_0^{\infty}  { \lambda_x}  \delta(t)e^{-\lambda_x t} dt > \int_0^{\infty} \lambda_x e^{(\delta'(0)-\lambda_x) t} dt = \frac{\lambda_x}{\lambda_x - \delta'(0)}.
 \]
Denote $c_x : = \frac{\lambda_x}{\lambda_x - \delta'(0)}$.

If $S=\{x\}$, then as $x\neq 0$ by Lemma \ref{L12} we have that
$$\sum_{y\neq x}J(y,S)q_{xy}\leq x\sum_{y\neq x}\E_y[\delta(T_y)]q_{xy}<x\lambda_x<x(\lambda_x-\delta'(0)),$$
which implies that $S$ is a strong equilibrium.

For the rest of the proof, we assume $S$ contains at least two points. Fix any $x \in S$, we have
$$J(x, \hat S)=\sum_{y \in S\backslash \{x\}}\frac{q_{xy}}{\lambda_x}\E_x[\delta(\tau_{\hat S}) X_{\tau_{\hat S}}|X_{T_x} = y]+\sum_{y \notin S}\frac{q_{xy}}{\lambda_x} \E_x[\delta(\tau_{\hat S}) X_{\tau_{\hat S}}|X_{T_x} = y].$$
Since for $y\in S\setminus\{x\}$,
$$\E_x[\delta(\tau_{\hat S}) X_{\tau_{\hat S}}|X_{T_x} = y]=y\E_x[\delta(\tau_{\hat S})|X_{T_x} = y]=y\E_x[\delta(T_x)|X_{T_x} = y]=y\E_x[\delta(T_x)],$$
and for $y\in S^c$,
\begin{align*}
\E_x[\delta(\tau_{\hat S}) X_{\tau_{\hat S}}|X_{T_x} = y]&\geq \E_x[\delta(T_x)\delta(\tau_{\hat S}-T_x) X_{\tau_{\hat S}}|X_{T_x} = y]\\
&=\E_x[\delta(T_x)|X_{T_x} = y]\cdot\E_x[\delta(\tau_{\hat S}-T_x) X_{\tau_{\hat S}}|X_{T_x} = y]=\E_x[\delta(T_x)]\cdot J(y,\hat S),
\end{align*}
we have that
\begin{equation}\label{e9}
J(x,\hat S)\geq\left(\sum_{y \in S\backslash \{x\}}\frac{q_{xy}}{\lambda_x}  y+\sum_{y \notin S}\frac{q_{xy}}{\lambda_x} J(y, \hat S)\right)\cdot\E_x[\delta(T_x)].
\end{equation}
Denote
\[
 \text{I} : =   \sum_{y \in S\backslash \{x\}} \frac{q_{xy}}{\lambda_x}y, \quad  \text{II}  : = \sum_{y \notin S}\frac{q_{xy}}{\lambda_x} J(y, S), \quad \hat{ \text{II} } : = \sum_{y \notin S}\frac{q_{xy}}{\lambda_x}J(y, \hat S).
 \]
 By Lemma \ref{L12}, $y> 0$ for all $y \in \hat S$ and $J(y, \hat S)> 0$ for all $y \notin \hat S$, thus I + $\hat{\text{II}} > 0$. This together with $\E_x[\delta(T_x)] > c_x$ implies that
 \[
 J(x, \hat S)> ( \text{I } + \hat{ \text{II} })c_x.
 \]
 Then
 \begin{align*}
 x - J(x, \hat S) & < x - ( \text{I } + \hat{ \text{II} })c_x \\
 & = x - (\text{I } +   \text{II})c_x + ( \text{II} - \hat{  \text{II}})c_x \\
 & = x -  ( \text{I } +  \text{II})c_x + c_x  \sum_{y \notin S}\frac{q_{xy}}{\lambda_x}(J(y, S) - J(y, \hat S))\\
 & \le  x -  ( \text{I }+   \text{II})c_x + c_x  \sum_{y \notin S}\frac{q_{xy}}{\lambda_x}(\E_y[\delta(\tau_{S})\textbf{1}_{\{X_{\tau_S} = x\}}](x -J(x, \hat S)),
 \end{align*}
where the last line follows from Lemma \ref{L11}.
Thus 
\begin{equation}\label{e10}
\left(1 - c_x  \sum_{y \notin S}\frac{q_{xy}}{\lambda_x}(\E_y[\delta(\tau_{S})\textbf{1}_{\{X_{\tau_S} = x\}}])\right)(x -J(x, \hat S)) <  x -  ( \text{I}+  \text{II})c_x. 
\end{equation}
Notice that 
\[
c_x  \sum_{y \notin S}\frac{q_{xy}}{\lambda_x}(\E_y[\delta(\tau_{S})\textbf{1}_{\{X_{\tau_S} = x\}}] \le c_x  \sum_{y \notin S}\frac{q_{xy}}{\lambda_x} \le c_x < 1.
\]
Then by Lemma \ref{L12},
$$x -  ( \text{I}+  \text{II})c_x > 0,\quad\forall\,x\in\mathbb{X},$$
which implies $S$ is a strong equilibrium.
 \end{proof}

\subsection{Proof of Theorem \ref{Existence}}

We start with the following lemma, which in particular indicates that a smaller mild equilibrium generates larger values.

\begin{lemma} Let Assumption \ref{A2} (i) hold. If $S$ is a mild equilibrium, then for any subset $R \subset \mathbb{X}$ with $S \subset R$, we have\[
J(x, S) \ge J(x, R), \quad \forall x \in \mathbb{X}.
\]
\end{lemma}
\begin{proof}
Since $S\subset R$, $\rho(x, S)\ge \rho(x, R)$ for all $x\in \mathbb{X}$.
\begin{align*}
J(x, S)& = \E_x[\delta(\rho(x, S)) X_{\rho(x, S)}]\\
& = \E_x[\E_x[\delta(\rho(x, S)) X_{\rho(x, S)} |\mathcal{F}_{\rho(x, R)} ]]\\
& \ge \E_x[\delta(\rho(x, R))\E_x[\delta(\rho(x, S)- \rho(x, R)) X_{\rho(x, S)} |\mathcal{F}_{\rho(x, R)} ]]\\
& = \E_x[\delta(\rho(x, R))\E_{X_{\rho(x, R)}}[\delta(\rho(X_{\rho(x, R)}, S)) X_{\rho(x, S)} ]]\\
& \ge \E_x[\delta(\rho(x, R))X_{\rho(x, R)}] = J(x, R).
\end{align*}
The last inequality holds because $S$ is a mild equilibrium and by definition, \[\E_{X_{\rho(x, R)}}[\delta(\rho(X_{\rho(x, R)}, S)) X_{\rho(x, S)} ] \ge X_{\rho(x, R)}.
\]
\end{proof}

\begin{cor}\label{c1}
Let Assumption \ref{A2} (i) hold. If $S$ is the smallest mild equilibrium, i.e. $S\subset \widetilde S$ for any mild equilibrium $\widetilde S$, then $S$ is an optimal mild equilibrium.
\end{cor}
Thanks to this corollary, in order to show $S_\infty$ defined in \eqref{e7} is an optimal mild equilibrium, it suffices to show that $S_\infty$ is the smallest one.

Recall $S_n$ defined in \eqref{e8}. We have the following lemma.

\begin{lemma}\label{L21}
For any mild equilibrium $R$, we have that $S_n \subset R$ for all $n \in \mathbb{N}$.
\end{lemma}
\begin{proof}
We prove this lemma by induction. First $S_0 \subset R$. Suppose $S_n\subset R$ for $n \ge 0$. Since $R$ is a mild equilibrium, for any $x \notin R$, \[
x \le J(x, R) \le  \sup_{S: S_n \subset S \subset \mathbb{X}\backslash \{x\}} J(x, S).
\]
Therefore $x\notin S_{n+1}$. As a result, $S_{n+1} \subset R$ for all $n \in  \mathbb{N}$.
\end{proof}
\begin{lemma}\label{L22}
Let  Assumption \ref{A1} (i) hold. For $y \notin S_{\infty}$, denote \[V_n := \sup_{S: S_n \subset S \subset \mathbb{X}\backslash \{y\}} J(y, S), \quad \quad V_{\infty} := \sup_{S: S_{\infty} \subset S \subset \mathbb{X}\backslash \{y\}} J(y, S),\] then we have $V_n \searrow V_{\infty}, n \to \infty$.
\end{lemma}
\begin{proof}
Since $S_{\infty}= \bigcup_{n\ge 1}S_n$, we have $\rho(y, S_{\infty}\backslash S_n) \to \infty, n \to \infty$. Then for any $\varepsilon >0$, there exists $N = N(\varepsilon, y)$ such that for $n > N$, $\E_y[\delta(\tau_{S_{\infty}\backslash S_n})] < \varepsilon$ since $\lim_{t \to \infty}\delta(t) = 0$.

For any $R_n$ such that $S_n \subset R_n \subset \mathbb{X}\backslash\{y\}$, denote $\overline{R_n} : = R_n \bigcup S_{\infty}$, then we have,
\begin{align*}
    J(y, R_n) - J(y, \overline{R_n}) & = \E_y[(\delta(\tau_{R_n})X_{\tau_{R_n}} - \delta(\tau_{\overline {R_n}})X_{\tau_{\overline{R_n}}} )\textbf{1}_{\{X_{\tau_{\overline{R_n}}} \in S_{\infty}\backslash R_n\}}] \\
    & \le C\E_y[\delta(\tau_{R_n})\textbf{1}_{\{X_{\tau_{\overline{R_n}}} \in S_{\infty}\backslash R_n\}}] \\
    & \le C\E_y[\delta(\tau_{S_{\infty}\backslash R_n})\textbf{1}_{\{X_{\tau_{\overline{R_n}}} \in S_{\infty}\backslash R_n\}}] \\
    & \le C\varepsilon
\end{align*}

Since $S_{\infty} \subset \overline{R_n} \subset \mathbb{X}\backslash\{y\}$, by definition, $J(y, \overline{R_n}) \le V_{\infty}$. Therefore we have that for any $\varepsilon>0$, there exists $N$ such that for any $n \ge N$,\[
V_n = \sup_{R_n: S_n \subset R_n \subset \mathbb{X}\backslash\{y\}}J(y, R_n) \le V_{\infty} + C\varepsilon.
\]
Clearly $S_n \subset S_{n+1}$ implies that $V_n$ is non-increasing and $V_n \ge V_{\infty}$ for all $n$. This completes the proof that $V_n \searrow V_{\infty}, n \to \infty$.
\end{proof}

\begin{proof}[\textbf{Proof of Theorem \ref{Existence}}]
By Corollary \ref{c1} and Lemma \ref{L21}, to show that $S_{\infty}$ is an optimal mild equilibrium, it suffices to show $S_\infty$ is a mild equilibrium.

Suppose $S_{\infty}$ is not a mild equilibrium. Then \[
\alpha := \sup_{x\in \mathbb{X}}\{x -J(x, S_{\infty})\} >0.
\]
For any $\varepsilon>0$, there exists $y \notin S_{\infty}$ such that $y - J(y, S_{\infty}) \ge \alpha- \varepsilon$.
Since $y \notin S_n$ for all $n \ge 0$, we have \[
y \le \sup_{S: S_n \subset S \subset \mathbb{X}\backslash \{y\}} J(y, S), \quad \forall n\ge 0.
\]
By Lemma \ref{L22}, \[
y \le \sup_{S: S_{\infty} \subset S \subset \mathbb{X}\backslash \{y\}} J(y, S).\]
Thus, there exists subset $R$ with $S_{\infty} \subset R \subset \mathbb{X}\backslash \{y\}$ such that \[
y \le J(y, R)+\varepsilon.
\]
Then we have $J(y, R)- J(y, S_{\infty}) \ge y- \varepsilon+ \alpha - \varepsilon -y = \alpha-2\varepsilon$. Since $S_{\infty} \subset R$, $\rho(y, S_{\infty}) \ge \rho(y, R)$. It follows that
\begin{align*}
 J(y, R)- J(y, S_{\infty})
=\,\,&  \E_y[\delta(\rho(y, R)) X_{\rho(y, R)}] -  \E_y[\E_y[\delta(\rho(y, S_{\infty})) X_{\rho(y, S_{\infty})}| \mathcal{F}_{\rho(y, R)}] ]\\
\le \,\,&  \E_y[\delta(\rho(y, R)) X_{\rho(y, R)}] -  \E_y[\delta(\rho(y, R))\E_y[\delta(\rho(y, S_{\infty})-\rho(y, R)) X_{\rho(y, S_{\infty})}| \mathcal{F}_{\rho(y, R)}] ]\\
= \,\,&  \E_y[\delta(\rho(y, R)) (X_{\rho(y, R)}- \E_{X_{\rho(y, R)}}[\delta(\rho(X_{\rho(y, R)}, S_{\infty}))X_{\rho(X_{\rho(y, R)}, S_{\infty})}]) \\
\le \,\,& \E_y[\delta(\rho(y, R))]\alpha\\
 \le \,\,&  \E_y[\delta(T_y)]\alpha. 
\end{align*}
By Assumption \ref{A2} (i), $\lambda = \sup_{x\in \mathbb{X}}\lambda_x < \infty$ and since $y \notin R$, we have $0< \E_y[\delta(T_y)]<c<1$ where $c= \int_0^{\infty}\delta(t) \lambda e^{-\lambda t} dt$.
By choosing $0< \varepsilon \le \frac{\alpha(1- c)}{2}$, we obtain a contradiction.

Next let us prove $S_\infty$ is a strong equilibrium. If $X$ is irreducible, then $S_\infty$ is a strong equilibrium by Theorem \ref{Connection}. In general, following the proof for Proposition \ref{p1}, to show $S_\infty$ is a strong equilibrium, it suffices to show that for any $x\in S_\infty$ with $\lambda_x>0$,
\begin{equation}\label{e12}
x(\lambda_x - \delta'(0)) > \sum_{y\in S_\infty\backslash \{x\}}y q_{xy} + \sum_{y \in S_\infty^c} \E_y[\delta(\tau_S)X_{\tau_S}]q_{xy}.
\end{equation}
Take $x\in S_\infty$ with $\lambda_x>0$. Following the argument for \eqref{e9}, we have that
$$J(x,\hat S_\infty)\geq\left(\sum_{y \in S_\infty\backslash \{x\}}\frac{q_{xy}}{\lambda_x}  y+\sum_{y \notin S_\infty}\frac{q_{xy}}{\lambda_x} J(y, \hat S_\infty)\right)\cdot\E_x[\delta(T_x)],$$
where $\hat S_\infty=S_\infty\setminus\{x\}$. Using an argument similar to that for \eqref{e10}, we have that
$$\left(1 - c_x  \sum_{y \notin S_\infty}\frac{q_{xy}}{\lambda_x}(\E_y[\delta(\tau_{S_\infty})\textbf{1}_{\{X_{\tau_{S_\infty}} = x\}}])\right)(x -J(x, \hat S_\infty)) \leq  x -  ( \text{I}_\infty+  \text{II}_\infty)c_x,$$
where
$$\text{I}_\infty : =   \sum_{y \in S_\infty\backslash \{x\}} \frac{q_{xy}}{\lambda_x}y\quad\text{and}\quad\text{II}_\infty  : = \sum_{y \notin S_\infty}\frac{q_{xy}}{\lambda_x} J(y, S_\infty).$$
Since $S_\infty$ is the smallest mild equilibrium, $\hat S_\infty$ is not a mild equilibrium. Then $x>J(x,\hat S_\infty)$ by \eqref{e11}. Therefore,
$$x -  ( \text{I}_\infty+  \text{II}_\infty)c_x>0,$$
which implies \eqref{e12}.
\end{proof}

 {
\section{Examples illustrating the iteration method in Theorem \ref{Existence}}

In this section, we provide examples to demonstrate the iteration method in Theorem \ref{Existence}.

The next proposition shows that the iteration method in Theorem \ref{Existence} will terminate within one step in the case of time consistency and leads to an optimal stopping time.

\begin{prop}\label{prop1} If $\delta(s)\delta(t) = \delta(s+t)$ for all $t, s \ge 0$ and Assumptions \ref{A1} (i) holds, then $S_1 = S_n$ for all $n \ge 2$ and $S_1$ is an optimal stopping strategy.
\end{prop}

\begin{proof} By definition, $S_0 = \emptyset$ and \[
S_1 =\{ x \in \mathbb{X}: x > \sup_{S\subset \mathbb{X}\backslash \{x\}} J(x, S)\}.
\] 

We show that for any $x \notin S_1$ and any set $R \subset \mathbb{X}\backslash \{x\}$, we have $J(x, R) \le J(x, \tilde R)$ where $\tilde R = R \cup S_1 = R\cup (S_1 \backslash R)$. 

For $S_1 \subset R$, $R = \tilde R$ and $J(x, R) \le J(x, \tilde R)$  holds trivially.

For $S_1 \not\subset R$, denote $\gamma = \tau_R$ and $\tilde \gamma = \tau_{\tilde R}$. Then a.s. $\gamma \ge \tilde \gamma$. We have 
\begin{align*}
J(x, R) - J(x, \tilde R) & = \E_x[\delta(\gamma)X_{\gamma}] -  \E_x[\delta(\tilde\gamma)X_{\tilde\gamma}]\\
& =  \E_x[\delta(\gamma)X_{\gamma}(\mathbf{1}_{\{X_{\tilde\gamma} \in R\}} +\mathbf{1}_{\{X_{\tilde\gamma} \notin R\}} )] -  \E_x[\delta(\tilde\gamma)X_{\tilde\gamma}]\\
& = \E_x[(\delta(\gamma)X_{\gamma} - \delta(\tilde\gamma)X_{\tilde\gamma})\mathbf{1}_{\{X_{\tilde\gamma} \notin R\}}]\\
& = \E_x[\E[(\delta(\gamma)X_{\gamma} - \delta(\tilde\gamma)X_{\tilde\gamma})\mathbf{1}_{\{X_{\tilde\gamma} \notin R\}} | \mathcal{F}_{\tilde \gamma}]]\\
& = \E_x[\E_{X_{\tilde \gamma}}[\delta(\gamma)X_{\gamma}] \delta({\tilde\gamma})\mathbf{1}_{\{X_{\tilde\gamma} \notin R\}} - \delta(\tilde\gamma)X_{\tilde\gamma}\mathbf{1}_{\{X_{\tilde\gamma} \notin R\}} ]\\
& = \E_x[ \delta(\tilde\gamma)\mathbf{1}_{\{X_{\tilde\gamma} \notin R\}} ( \E_{X_{\tilde \gamma}}[\delta(\gamma)X_{\gamma}]-X_{\tilde\gamma})]\\
& \le 0,
\end{align*}
since on $\{ X_{\tilde\gamma} \notin R\}$, $X_{\tilde\gamma} \in S_1$ and $X_{\tilde\gamma} >\E_{X_{\tilde \gamma}}[\delta(\gamma)X_{\gamma}]$.

As a result, $J(x, R) \le \sup_{S: S_1 \subset S \subset \mathbb{X}\backslash \{x\}} J(x, S)$ for all $R \subset \mathbb{X}\backslash \{x\}$. Thus for any $x \not\in S_1$, $x \le \sup_{R: R\subset \mathbb{X}\backslash \{x\}} J(x, R) \le  \sup_{S: S_1 \subset S \subset \mathbb{X}\backslash \{x\}} J(x, S)$, which implies $S_1  = S_2 = S_{\infty}$. 

Next we show that for all $x\in \mathbb{X}$, \[
J(x, S_1) \ge J(x, S), \quad \quad \forall S \subset \mathbb{X}.
\]

For $x \in S_1$, by definition of $S_1$,  $x > \sup_{S\subset \mathbb{X}\backslash \{x\}} J(x, S)$. Therefore $x \ge \sup_{S\subset \mathbb{X}} J(x, S)$.

For $x \not\in S_1$, for any $S\subset \mathbb{X}$, let $\tilde S = S\cup S_1$.
Since $S \subset \tilde S$ and $S_1 \subset \tilde S$, we have a.s. $\tau_{S} \ge \tau_{\tilde S}$ and $\tau_{S_1} \ge \tau_{\tilde S}$. By similar arguments as above, we obtain
\[
J(x, \tilde S) - J(x, S) = \E_x[\delta(\tau_{\tilde S})\mathbf{1}_{\{X_{\tilde S} \notin S\}}(X_{\tau_{\tilde S}} - \E_{X_{\tau_{\tilde S}}}[\delta(\tau_S)X_{\tau_S}])] \ge 0,
\] 
and \[
J(x, S_1) - J(x, \tilde S)  = \E_x[\delta(\tau_{\tilde S})\mathbf{1}_{\{X_{\tilde S} \notin S_1\}}(\E_{X_{\tau_{\tilde S}}}[\delta(\tau_{S_1})X_{\tau_{S_1}}]- X_{\tau_{\tilde S}} )] \ge 0.
\]

Therefore $J(x, S_1) \ge J(x, \tilde S) \ge J(x, S)$.
\end{proof}

In the case of time inconsistency, the above result generally does not hold. The next example demonstrates an application of the iteration method in Theorem \ref{Existence}.

\begin{Example}\label{E1}
Consider hyperbolic discount function $\delta(t) = \frac{1}{1+ \beta t}$ for $\beta >0$ and  $\mathbb{X} = \{x_1, x_2, x_3, x_4\}$, whose generator is given by 
\[
Q = \begin{bmatrix}
-\lambda_1 & q_{12} & q_{13} & q_{14} \\
 q_{21} & -\lambda_2 & q_{23} & q_{24} \\
q_{31} & q_{23} & -\lambda_3 &  q_{34} \\
q_{41} & q_{42} & q_{43} & -\lambda_4
\end{bmatrix}
= \begin{bmatrix}
-3 & 1 & 1 & 1 \\
0 & -1 & 0 & 1 \\
0 & 0.4 & -2 & 1.6 \\
1 & 1 & 1 & -3
\end{bmatrix}.
\]

Let $\beta = 3$, $x_1 = 10, x_2 = 40, x_3 = 46, x_4 = 100$. 

Next we show that by applying the iteration method, we have $S_0 = \emptyset, S_1 = \{x_2, x_4\}, S_2 = \{x_2, x_3, x_4\} = S_{\infty}$. 

Denote $T_i := \inf\{t \ge 0: X_t \ne x_i| X_0 = x_i\}$.

(i) Since $100 = x_4 > x_3 = 46 \ge \sup_{S\subset \mathbb{X}\backslash \{x_4\}} J(x_4, S)$, we have that $x_4 \in S_1$.

(ii) For $x_3$, consider $S = \{x_4\}$. \begin{align*}J(x_3, \{x_4\}) &= \E_{x_3}[\delta(\tau_{\{x_4\}})X_{\tau_{\{x_4\}}}] \\
&= x_4(\frac{q_{34}}{\lambda_3}\E_{x_3}[\delta(T_3)] +\frac{q_{32}}{\lambda_3} \E_{x_3}[\E[\delta(T_2+T_3) | X_{T_3} = x_2]])\\
& = 100(0.8\int_0^{\infty} \frac{2}{1+3t} e^{-2t}dt + 0.2\int_0^{\infty}\int_0^{\infty} \frac{2}{1+3(t+s)} e^{-t-2s} dt ds)\\
& \doteq 100 (0.8\times 0.5173 + 0.2 \times 0.2539) = 46.46.
\end{align*}
Therefore $x_3 = 46 < J(x_3, \{x_4\}) \le \sup_{S\subset \mathbb{X}\backslash \{x_3\}} J(x_3, S) $ and $x_3 \notin S_1$.

(iii) Note that $\sup_{S\subset \mathbb{X}\backslash \{x_2\}} J(x_2, S) \le x_4 \E_{x_2}[\delta(T_2)]$. We have \[
\E_{x_2}[\delta(T_2)] = \int_0^{\infty} \frac{1}{1+3t}e^{-t} dt  \doteq  0.3856.
\]
Therefore $x_2 = 40 > 0.3856 \times 100 = x_4 \E_{x_2}[\delta(T_2)] \ge \sup_{S\subset \mathbb{X}\backslash \{x_2\}} J(x_2, S) $ and $x_2 \in S_1$.

(iv) For $x_1$, consider $S = \{x_2, x_3, x_4\}$. \begin{align*}J(x_1, \{x_2, x_3, x_4\}) &= \E_{x_1}[\delta(T_1)X_{\tau_{\{x_2, x_3, x_4\}}}] \\
&= \frac{1}{3}(x_2+x_3+x_4)\E_{x_1}[\delta(T_1)]\\
& = 62\times\int_0^{\infty} \frac{3}{1+3t}e^{-3t} dt \doteq 62\times 0.5963.
\end{align*}
Thus $x_1 = 10 \le J(x_1, \{x_2, x_3, x_4\}) \le \sup_{S\subset \mathbb{X}\backslash \{x_1\}} J(x_1, S)$ and $x_1 \notin S_1$.

(v) By (iv), $x_1 \notin S_2$ given that $S_1 = \{x_2, x_4\}$.

(vi) To show that $x_3 \in S_2$, we only need to show that $J(x_3, \{x_2, x_4\}) < x_3$ since $S_1 = \{x_2, x_4\}$ and $q_{31} = 0$.
\begin{align*}
J(x_3, \{x_2, x_4\})&= (\frac{q_{32}}{\lambda_3}x_2+ \frac{q_{34}}{\lambda_3}x_4)\E_{x_3}[\delta(T_3)]\\
& = (0.2\times 40 + 0.8 \times 100)\int_0^{\infty} \frac{2}{1+3t}e^{-2t}dt\\
&\doteq 88\times 0.5173  = 45.52
\end{align*}
Thus $x_3 \in S_2$.

(vii) Again by (iv), $x_1 \notin S_3$ given that $S_2 = \{x_2, x_3, x_4\}$. Therefore $S_n = S_2 = \{x_2, x_3, x_4\}$ for $n \ge 2$.

\end{Example}

}

{ 

\subsection{Example 3.2}\label{E2} In this example, process $X$ has infinite state space and can be viewed as the payoff of some American option. Consider a stock price process $Y$ that takes values in $\mathbb{Y} := \{ u^i: i \in \mathbb{Z}\}$ for some fixed $u > 1$. There exists $\lambda > 0$ and $p \in [\frac{1}{1+u}, 1)$ such that\[
q_{u^i u^{i+1}} = p\lambda, \quad  q_{u^i u^{i-1}} = (1-p)\lambda, \quad \forall i  \in \mathbb{Z}.
\]

Let  the discount function be $\delta(t) = 
\frac{1}{1+\beta t}$ for some constant $\beta > 0$ and let the payoff process be $X = f(Y)$ for some payoff function $f(y) = (K - y)^+$, where $K$ is a positive constant. Since $f$ is bounded and nonnegative, our results still holds when we have $X = f(Y)$. Next we will show how to use the iteration method to find an optimal mild equilibrium in this problem.

\begin{lemma} \label{L0} $S_1 = \{u^i \in (0, K): K- u^i > J(u^i, \{u^m\}), \forall m < i, m\in \mathbb{Z},  i \in \mathbb{Z}\}$.
\end{lemma}
\begin{proof} 
Since for any $u^i \ge K$, $f(u^i) = (K-u^i)^+ = 0 \le J(u^i, (0, K) \cap \mathbb{Y})$, we obtain that $S_1 \subset (0, K)$. Thus we only consider $u^i \in (0, K)$ and we show that $u^i \in S_1$ if and only if  $u^i \in (0, K)$ and  $K - u^i > J(u^i, \{u^l\})$ for all $l < i$.

``$\Longrightarrow$'': Take $u^i\in S_1$. Then obviously $K - u^i > J(u^i, \{u^l\})$ for all $l < i$.
 
``$\Longleftarrow$'': Take $u^i \in (0, K)$. For any nonempty set $S\subset \mathbb{Y}\backslash \{u^i\}$, there are three cases.

Case 1: $S\in A:=\{\tilde S\subset \mathbb{Y}\backslash \{u^i\}: \tilde S \cap (0, u^i) = \emptyset\text{ and }\tilde S \cap (u^i, \infty) \ne \emptyset\}$. Let $u^r = \min (S \cap (u^i, \infty))$. Then $J(u^i, S) = J(u^i, u^r)$. Since $u^r > u^i$, we have $f(u^i) = K - u^i > J(u^i, \{u^r\}) = (K-u^r)^+\E_{u^i}[\delta(\tau_{\{u^r\}})]$. Then obviously we have that
$$f(u^i)>\sup_{S\in A}J(u^i,S).$$

Case 2: $S\in B:=\{\tilde S\subset \mathbb{Y}\backslash \{u^i\}: \tilde S \cap (0, u^i) \neq \emptyset\text{ and }\tilde S \cap (u^i, \infty) = \emptyset\}$. Note that for any $n\in\mathbb{Z}$ such that $n<i$,
$$K-u^i>\sup_{n\leq k\leq i-1}J(u^i,\{u^k\}).$$
Moreover, $\lim_{n\to\-\infty}J(u^i,\{u^n\})=0$. Thus
\begin{equation}\label{ee102}
K-u^i>\sup_{k\leq i-1}J(u^i,\{u^k\}).
\end{equation}
Now let $u^l = \max (S \cap (0, u^i))$. Then $J(u^i, S) = J(u^i, \{u^l\})$. Thus by \eqref{ee102},
$$f(u^i)>\sup_{S\in B}J(u^i,S).$$

Case 3: $S\in C:=\{\tilde S\subset \mathbb{Y}\backslash \{u^i\}: \tilde S \cap (0, u^i) \neq \emptyset\text{ and }\tilde S \cap (u^i, \infty) \neq \emptyset\}$. Let $u^l = \max (S \cap (0, u^i))$ and $u^r = \min (S \cap (u^i, \infty))$. 

If $u^r \le K$, observe that $Y$ is a submartingale, so $u^i \le \E_{u^i}[Y_{\tau_{\{u^l, u^r\}}}] = \E_{u^i}[Y_{\tau_S}]$. Thus $f(u^i) = K - u^i \ge\E_{u^i}[(K - Y_{\tau_S})^+] = \E_{u^i}[(K - Y_{\tau_S})^+]> \E_{u^i}[\delta(\tau_{\{u^l, u^r\}})(K - Y_{\tau_S})^+]$.

If $u^r > K$, then 
\begin{align*}
J(u^i, S) & = \E_{u^i}[\delta(\tau_{\{u^l\}})(K - u^l)^+ \mathbf{1}_{\{\tau_{\{u^l\}} < \tau_{\{u^r\}}\}}]  + \E_{u^i}[\delta(\tau_{\{u^r\}})(K - u^r)^+ \mathbf{1}_{\{\tau_{\{u^r\}} < \tau_{\{u^l\}}\}}]\\
& =  \E_{u^i}[\delta(\tau_{\{u^l\}})(K - u^l)\mathbf{1}_{\{\tau_{\{u^l\}} < \tau_{\{u^r\}}\}} ]\\
& \le  \E_{u^i}[\delta(\tau_{\{u^l\}})(K - u^l)] \\
& = J(u^i, \{u^l\})\\
& \leq \sup_{S\in B}J(u^i,S).
\end{align*}

Therefore,
$$f(u^i)>\sup_{S\in C}J(u^i,S).$$
This completes the proof.
\end{proof}

Fix $m, i \in \mathbb{Z}$ such that $m < i  < \log_u K$. 
$J(u^i, \{u^m\}) = (K-u^m)\E_{u^i}[\delta(\tau_{\{u^m\}})]$. 
Since $(q_{u^i u^{j}})_{j \ne i}$ are the same for each $i \in \mathbb{Z}$, we have $\E_{u^i}[\delta(\tau_{\{u^m\}})] = \E_{u^{i-k}}[\delta(\tau_{\{u^{m-k}\}})]$  for any $k \in \mathbb{N}$. Therefore denote $\alpha_{i-m}: =\E_{u^i}[\delta(\tau_{\{u^m\}})]$. Note that $\alpha_n, n\in \mathbb{N}$ can be computed explicitly.
 For example,\[
\alpha_1 = \sum_{k = 1}^{\infty} \frac{\binom{2k-1}{k}p^{k-1}(1-p)^k}{2k-1}\cdot\int_0^{\infty}\frac{1}{1+\beta t}g(t, 2k-1)dt,
\]
where $g(t,n) = \frac{\lambda^n}{(n-1)!} t^{n-1}e^{-\lambda t}$ is the density function of gamma distribution with shape parameter $n$ and rate parameter $\lambda$.

\begin{prop}\label{prop2} $S_{\infty} = \{u^i: i \le n_0\}$ where $n_0 = \lceil \log_u (\frac{1-\alpha_1}{u-\alpha_1}K) \rceil$.
\end{prop}
\begin{proof} Since for any $u^i \ge K$ and any $S \subset \mathbb{Y}\backslash \{u^i\}$, $f(u^i) = 0 \le \sup J(u^i, S)$, $S_{\infty} \subset (0, K)\cap \mathbb{Y}$. In the following we only consider $u^i$ with $i  \le  \lfloor \log_u K \rfloor$. Consider sequence $\{\frac{K - u^m}{K}\}_{m \le  \lfloor \log_u K \rfloor}$. It is easy to check that \[
\frac{K - u^{m-1}}{K} > \frac{K - u^m}{K} > 0, \quad \forall m \le  \lfloor \log_u K \rfloor,
\]
and $\lim_{m \to -\infty} \frac{K - u^{m}}{K} = 1$. Then there exists $m_0 \le \lfloor \log_u K \rfloor$ such that \[
\frac{K - u^{m_0}}{K} > \alpha_1\ge \frac{K - u^{m_0+1}}{K}.
\]
Then $K- u^{m_0} > K \alpha_1 > (K- u^{m}) \alpha_{m_0-m}, \,\, \forall m < m_0$. By Lemma \ref{L0}, $u^{m_0} \in S_1$.
Since $K- u^{m} > K- u^{m_0} > K \alpha_1$ for all $m < m_0$, by similar argument, $u^{m} \in S_1, \forall m < m_0$. Therefore \[
\{u^m: m \le m_0\} \subset S_1.
\]
Consider the sequence $\{\frac{K - u^n}{K- u^{n-1}}\}_{n \le  \lfloor \log_u K \rfloor}$. It is easy to check that \[
 \frac{K - u^{n-1}}{K- u^{n-2}} > \frac{K - u^n}{K- u^{n-1}} \geq 0, \quad \forall n \le  \lfloor \log_u K \rfloor,
\]
and $\lim_{n \to -\infty} \frac{K - u^{n}}{K- u^{n-1}} = 1$. Then there exists $n_0 \le \lfloor \log_u K \rfloor$ such that
\begin{equation}\label{ee101}
\frac{K - u^{n_0}}{K- u^{n_0-1}} > \alpha_1 \ge \frac{K - u^{n_0+1}}{K- u^{n_0}}.
\end{equation}
Then for any $n \ge n_0+1$, $(K- u^{n})^+\le (K- u^{n-1})^+ \alpha_1$. Thus $u^n \not\in S_1, \,\, \forall n > n_0$. That is \[
\{u^m: m \le m_0\} \subset S_1 \subset \{u^m: m \le n_0\} 
\]

Next we claim that for all $n \in \mathbb{N}$, $S_{n} \subset \{u^m: m \le n_0\}$. We will prove this claim by induction. By the above discussion, this claim holds for $n = 1$. Suppose $S_{n} \subset \{u^m: m \le n_0\}$ for $n \ge 1$. Then 
for any $m > n_0$, 
\[(K - u^{m})^+ \le (K- u^{m-1})^+ \alpha_1 \le \sup_{S: S_n \subset S \subset \mathbb{Y}\backslash\{u^{m}\} } J(u^{m}, S),\]
 which implies $u^m \not\in S_{n+1}$ for all $m > n_0$ and $S_{n+1} \subset \{u^m: m \le n_0\}$. As a result, $S_{\infty} \subset \{u^m: m \le n_0\} $.

If $m_0 = n_0$, then we have $S_1 =  \{u^m: m \le n_0\}  \subset S_{\infty}$. Thus $S_{\infty} =  \{u^m: m \le n_0\}$.

If $m_0 < n_0$, let $k = n_0- m_0$. Consider $u^{m_0+i}, i \in \{0, 1, 2, \cdots, k\}$. We claim that $u^{m_0+i} \in S_{i+1}$. Then we obtain $\{u^m: m \le n_0\}  \subset  S_{\infty}$.

Next we will prove this claim by induction. The claim holds for $i = 0$. Suppose $u^{m_0+i} \in S_{i+1}$. Then $ \{u^m: m \le m_0+i\} \subset S_{i+1}$. Consider the case when $u^{m_0+i+1} \not\in S_{i+1}$.
Note that
$$\sup_{S: S_{i+1} \subset S\subset \mathbb{Y}\backslash\{u^{m_0+i+1}\} } J(u^{m_0+i+1}, S)\leq\max_{m_0+i+2\leq k\leq n_0} J(u^{m_0+i+1},\{u^{m_0+i},u^k\})\vee J(u^{m_0+i+1},\{u^{m_0+i}\})$$
 As $Y$ is a submartingale, $K-u^{m_0+i+1} > J(u^{m_0+i+1},\{u^{m_0+i},u^k\})$ for any $k$ satisfying $m_0+i+2\leq k\leq n_0$. This together with \eqref{ee101} implies that
\[K-u^{m_0+i+1} >  \sup_{S: S_{i+1} \subset S\subset \mathbb{Y}\backslash\{u^{m_0+i+1}\} } J(u^{m_0+i+1}, S).\]
Thus, $u^{m_0+i+1} \in S_{i+2}$.

Therefore the iteration method will terminate within $n_0-m_0+1$ steps and we obtain $S_{\infty} =  \{u^m: m \le n_0\}$ where $n_0$ satisfies $\frac{K - u^{n_0}}{K- u^{n_0-1}} > \alpha_1 \ge \frac{K - u^{n_0+1}}{K- u^{n_0}}$. Equivalently, \[S_{\infty} = \left\{u^i: i \le  \left\lceil \log_u \left(\frac{1-\alpha_1}{u-\alpha_1}K\right) \right\rceil\right\}.\]
\end{proof}
}
\subsubsection{Discussion on how non-standard discounting affects the value of the option} 
%
Consider the optimal stopping problem
\begin{equation}\label{1001}
U(y):=\sup_{\tau\in\mathcal{T}}\E_y\left[\frac{1}{1+\beta\tau}(K-Y_\tau)^+\right].
\end{equation}
Let
\begin{equation}\label{1002}
\tau^*:=\inf\left\{t\geq 0:\ \delta(t)Y_t\geq\sup_{\tau\in\mathcal{T}_t}\E\left[\delta(\tau)(K-Y_\tau)^+\Big|\mathcal{F}_t\right]\right\},
\end{equation}
where $\mathcal{T}_t$ is the set of stopping times taking values in $[t,\infty]$. From the classical theory of optimal stopping we know that $\tau^*$ is an optimal solution for the problem \eqref{1001}. Recall that
$$\tau_{S_\infty}:=\inf\{t\geq 0:\ Y_t\in S_\infty\}$$
is the stopping time corresponding to the optimal mild equilibrium $S_\infty$, where $S_\infty$ is obtained from the iteration in Proposition \ref{prop2}. We have the following.

\begin{prop}\label{propdiscussion} Suppose $\log_u (\frac{1-\alpha_1}{u-\alpha_1}K)$ is not an integer. Then
\begin{equation}\label{1003}
\tau_{S_\infty}\leq\tau^*.
\end{equation}
\end{prop}

Note that $\tau^*$ is an optimal pre-commitment strategy. That is, it is a strategy which is carried out based on the initial preference, and the agent commits to this strategy over the whole planning horizon and ignores the change of her future preference. On the other hand, $\tau_{S_\infty}$ is an equilibrium strategy (sophisticated strategy) which incorporates the change of preference. To be more specific, by using strategy $\tau_{S_\infty}$ the agent seriously takes the possible change of her future preference into consideration, and works on consistent planning: a strategy such that once it is enforced over time, all her future selves have no incentive to deviate from it. Proposition \ref{propdiscussion} indicates that with the recognition of the change of preference, the agent would actually expedite the exercise of the American put option.

As $\tau_{S_\infty}$ may not be optimal for the problem \eqref{1001}, the use of the equilibrium strategy $\tau_{S_\infty}$ will lower the expected payoff, if such evaluation is based on the initial preference. However, when the change of future preference is considered, there is no unique/proper way to define the dynamically optimal expected payoff over time. In this case, the equilibrium strategy is carried out such that the agent's future selves will not regret the decision.  

\begin{proof}[\textbf{Proof of Proposition \ref{propdiscussion}}]
By the Markov property of $Y$, we can rewrite \eqref{1002} as
\begin{equation}
\tau^*=\inf\{t\geq 0:\ Y_t\in A_t\},
\end{equation}
where
$$A_t:=\left\{y\in\mathbb{Y}:\ y\geq\sup_{\tau\in\mathcal{T}}\E_y\left[\frac{1+\beta t}{1+\beta(t+\tau)}(K-Y_\tau)^+\right]\right\}.$$
It is easy to see that $A_t\subset A_0$ for any $t\geq 0$. We claim that $A_0\subset S_\infty$, which further implies that $A_t\subset S_\infty$ for $t\geq 0$ and thus \eqref{1003}.

Indeed, take $u^n\in A_0$. Obviously $u^n\in(0,K)$. Then we have that
$$K-u^n=U(u^n)\geq J(u^n,\{u^{n-1}\})=\alpha_1(K-u^{n-1}),$$
which implies that
$$n\leq \log_u \left(\frac{1-\alpha_1}{u-\alpha_1}K\right)+1.$$
By assumption, we have that
$$n\leq n_0,$$
which implies $u^n\in S_\infty$.
\end{proof}
\begin{remark}
The assumption in Proposition \ref{propdiscussion}, i.e., $\log_u \left(\frac{1-\alpha_1}{u-\alpha_1}K\right)$ not being an integer, is very weak, since for it holds for a.e. $u$ and $K$.
\end{remark}
\begin{remark}
Here $\tau^*$ considered in \eqref{1002} is the smallest optimal solution for the problem \eqref{1001}. If it is replaced by the largest optimal solution, then the assumption in Proposition \ref{propdiscussion} is not needed.
\end{remark}

\section{Exact Containment of Equilibria: Optimal Mild $\subsetneqq$ Strong $\subsetneqq$ Weak  $\subsetneqq$ Mild}
In this section, we will use an example to illustrate that a mild equilibrium may not be a weak equilibrium, a weak equilibrium may not be a strong equilibrium and a strong equilibrium may not be an optimal mild equilibrium. 

Consider a two-state continuous-time Markov chain $X_t \in \{a, b\}$ for $t\ge 0$. Assume $a>0, b>0$ and without loss of generality we assume   $a>b$. The generator is\[
Q = \begin{bmatrix}
-\lambda_a & \lambda_a \\
\lambda_b & -\lambda_b
\end{bmatrix},
\]
where $\lambda_a>0$ and $\lambda_b>0$.

There are four subsets of $\{a, b\}$. Clearly $S = \emptyset$ and $S = \{b\}$ cannot be mild equilibria and $ S= \{a,b\}$ is a mild equilibrium. Next, let's check when $S = \{a\}$ is a mild equilibrium.

By definition $S = \{a\}$ is a mild equilibrium if and only if \[
b \le a\E_b[\delta(T_b)] = a \int_0^{\infty} \delta(t) \lambda_b e^{-\lambda_b t} dt.
\]
Consider the following cases.

(i) If $\frac{b}{a} = \int_0^{\infty} \delta(t) \lambda_b e^{-\lambda_b t} dt<1$, then both $\{a\}$ and $\{a,b\}$ are optimal mild equilibria and thus both are strong equilibria.

(ii) If $\frac{b}{a} < \int_0^{\infty} \delta(t) \lambda_b e^{-\lambda_b t} dt<1$, then $\{a\}$ is the only optimal mild equilibrium, which is also a strong equilibrium. But the mild equilibrium $\{a, b\}$ may not be a weak equilibrium. For example, when $\frac{b}{a}< \frac{\lambda_b}{\lambda_b -\delta'(0)}<1 $, the second condition for weak equilibrium is violated at state $b$, thus it is not a weak equilibrium.

(iii) $\frac{\lambda_a}{\lambda_a-\delta'(0)}<1<\frac{a}{b}$ holds automatically since $a>b$ and $\delta'(0)<0$. If $\frac{\lambda_b}{\lambda_b -\delta'(0)} <\frac{b}{a} < \int_0^{\infty} \delta(t) \lambda_b e^{-\lambda_b t} dt<1$,  then $\{a,b\}$ is not an optimal mild equilibrium, but it is a weak equilibrium and also a strong equilibrium.

(iv) If $\frac{\lambda_b}{\lambda_b -\delta'(0)} = \frac{b}{a} <1$, $\{a,b\}$ is a weak equilibrium, but it may not be a strong equilibrium, i.e.  condition (\ref{ii2}) on strong equilibrium may not hold at state $b$. This can be shown by computing the related term of order $\varepsilon^2$.

Since \[
\PP(X_{\varepsilon} =a | X_0= b) = \lambda_b \varepsilon -\frac{\lambda_b^2 +\lambda_a\lambda_b}{2}\varepsilon^2 +o(\varepsilon^2),
\]
and \[
\PP(X_{\varepsilon} =b | X_0= b) =1- \lambda_b \varepsilon +\frac{\lambda_b^2 +\lambda_a\lambda_b}{2}\varepsilon^2 +o(\varepsilon^2),
\]
we have \,\, $b - \E_b[\delta(\varepsilon)X_{\varepsilon}]$
\begin{align*} 
&=b - \delta(\varepsilon)[a\PP(X_{\varepsilon} =a | X_0= b) + b\PP(X_{\varepsilon} =b | X_0= b) ]\\
&= b - (1+\delta'(0)\varepsilon +\frac{\delta''(0)}{2}\varepsilon^2 +o(\varepsilon^2))[a(\lambda_b \varepsilon -\frac{\lambda_b^2 +\lambda_a\lambda_b}{2}\varepsilon^2 +o(\varepsilon^2))+b(1- \lambda_b \varepsilon +\frac{\lambda_b^2 +\lambda_a\lambda_b}{2}\varepsilon^2 +o(\varepsilon^2))]\\
&= (b\lambda_b - a\lambda_b -b\delta'(0))\varepsilon + [b(\lambda_b\delta'(0)-\frac{\lambda_b^2 +\lambda_a\lambda_b}{2}-\frac{\delta''(0)}{2}) -a(\delta'(0)\lambda_b- \frac{\lambda_b^2 +\lambda_a\lambda_b}{2})]\varepsilon^2 +o(\varepsilon^2)
\end{align*}
Therefore when the first order term and the second order term respectively satisfy
\begin{equation}\label{e13}
b(\lambda_b-\delta'(0))-a\lambda_b = 0,
\end{equation}
and 
\begin{equation}\label{e14}
b(\lambda_b\delta'(0)-\frac{\lambda_b^2 +\lambda_a\lambda_b}{2}-\frac{\delta''(0)}{2}) -a(\delta'(0)\lambda_b- \frac{\lambda_b^2 +\lambda_a\lambda_b}{2})<0,
\end{equation}
$\{a, b\}$ is a weak equilibrium but not a strong equilibrium.
Using \eqref{e13}, \eqref{e14} can be simplified to 
\begin{equation}\label{eq:cond-se}
\lambda_a+\lambda_b < \frac{\delta''(0)- 2(\delta'(0)^2)}{-\delta'(0)}.
\end{equation}

An interesting case is when
$\delta(t) = \frac{1}{1+\beta t}$. Then \eqref{eq:cond-se} does not hold: $\delta'(0) = -\beta$ and $\delta''(0)= 2\beta^2$. In this case $\frac{\delta''(0)- 2(\delta'(0)^2)}{-\delta'(0)} = 0$, which contradicts $\lambda_a+\lambda_b >0$. That means if we have hyperbolic discount function, a weak equilibrium is always a strong equilibrium in the two-state setting.

But when $\delta(t) = (1+\beta t)^{-\frac{1}{2}}$, then it can easily be seen that \eqref{eq:cond-se} holds:  $\delta'(0) = -\frac{\beta}{2}, \delta''(0)= \frac{3}{4}\beta^2$ implies that when $0<\lambda_a+\lambda_b <\frac{\beta}{2}$ and $\frac{b}{a} =\frac{2\lambda_b}{2\lambda_b+\beta}$, \{a, b\} is a weak equilibrium but not a strong equilibrium. In this case, $\{a, b\}$ is not an optimal mild equilibrium.

\bibliographystyle{apacite}
\bibliography{reference}{}

\end{document}